\pdfoutput=1
\ifdefined\HEADER\else
    \documentclass[cleveref]{lipics-v2021}
    \hideLIPIcs
    \nolinenumbers
\let\HEADER=\relax
\usepackage{newunicodechar}
\usepackage{silence}
\newunicodechar{♢}{\tikz \node[inner sep=1.5,draw,diamond] {};}
\newunicodechar{☆}{\tikz \node[inner sep=1,draw,star,star point ratio=2] {};}
\newunicodechar{△}{\triangle}
\newunicodechar{⬜}{\kern 0.5pt\tikz \node[inner sep=1.7,draw,regular polygon,regular polygon sides=4] {};\kern 0.5pt}
\newunicodechar{◯}{\tikz[baseline=-3pt] \node[inner sep=1.7,draw,cloud,cloud puffs=4,cloud puff arc=190] {};}

\newunicodechar{⊥}{\bot}
\newunicodechar{•}{\item}
\newunicodechar{✓}{\checkmark}
\newunicodechar{✗}{\xmark}
\newunicodechar{…}{\dots}
\newunicodechar{≔}{\coloneqq}
\newunicodechar{⁻}{^-}
\newunicodechar{⁺}{^+}
\newunicodechar{₋}{_-}
\newunicodechar{₊}{_+}
\newunicodechar{ℓ}{\ell}
\newunicodechar{•}{\item}
\newunicodechar{…}{\dots}
\newunicodechar{≔}{\coloneqq}
\newunicodechar{≤}{\leq}
\newunicodechar{≥}{\geq}
\newunicodechar{≰}{\nleq}
\newunicodechar{≱}{\ngeq}
\newunicodechar{⊕}{\oplus}
\newunicodechar{⊗}{\otimes}
\newunicodechar{≠}{\neq}
\newunicodechar{¬}{\neg}
\newunicodechar{≡}{\equiv}
\newunicodechar{₀}{_0}
\newunicodechar{₁}{_1}
\newunicodechar{₂}{_2}
\newunicodechar{₃}{_3}
\newunicodechar{₄}{_4}
\newunicodechar{₅}{_5}
\newunicodechar{₆}{_6}
\newunicodechar{₇}{_7}
\newunicodechar{₈}{_8}
\newunicodechar{₉}{_9}
\newunicodechar{⁰}{^0}
\newunicodechar{¹}{^1}
\newunicodechar{²}{^2}
\newunicodechar{³}{^3}
\newunicodechar{⁴}{^4}
\newunicodechar{⁵}{^5}
\newunicodechar{⁶}{^6}
\newunicodechar{⁷}{^7}
\newunicodechar{⁸}{^8}
\newunicodechar{⁹}{^9}
\newunicodechar{∈}{\in}
\newunicodechar{∉}{\notin}
\newunicodechar{⊂}{\subset}
\newunicodechar{⊃}{\supset}
\newunicodechar{⊆}{\subseteq}
\newunicodechar{⊇}{\supseteq}
\newunicodechar{⊄}{\nsubset}
\newunicodechar{⊅}{\nsupset}
\newunicodechar{⊈}{\nsubseteq}
\newunicodechar{⊉}{\nsupseteq}
\newunicodechar{∪}{\cup}
\newunicodechar{∩}{\cap}
\newunicodechar{∀}{\forall}
\newunicodechar{∃}{\exists}
\newunicodechar{∄}{\nexists}
\newunicodechar{∨}{\vee}
\newunicodechar{∧}{\wedge}
\newunicodechar{ℝ}{\mathbb{R}}
\newunicodechar{ℕ}{\mathbb{N}}
\newunicodechar{𝔼}{\mathbb{E}}
\newunicodechar{𝔽}{\mathbb{F}}
\newunicodechar{ℤ}{\mathbb{Z}}
\newunicodechar{𝒪}{\mathcal{O}}
\newunicodechar{⌊}{\lfloor}
\newunicodechar{⌋}{\rfloor}
\newunicodechar{⌈}{\lceil}
\newunicodechar{⌉}{\rceil}
\newunicodechar{·}{\cdot}
\newunicodechar{∘}{\circ}
\newunicodechar{×}{\times}
\newunicodechar{↑}{\uparrow}
\newunicodechar{↓}{\downarrow}
\newunicodechar{→}{\rightarrow}
\newunicodechar{←}{\leftarrow}
\newunicodechar{⇒}{\Rightarrow}
\newunicodechar{⇐}{\Leftarrow}
\newunicodechar{↔}{\leftrightarrow}
\newunicodechar{⇔}{\Leftrightarrow}
\newunicodechar{↦}{\mapsto}
\newunicodechar{∅}{\varnothing}
\newunicodechar{∞}{\infty}
\newunicodechar{≅}{\cong}
\newunicodechar{≈}{\approx}
\newunicodechar{ℓ}{\ell}
\newunicodechar{𝟙}{\mathds{1}}
\newunicodechar{𝟘}{\mathds{0}}

\newunicodechar{α}{\alpha}
\newunicodechar{β}{\beta}
\newunicodechar{γ}{\gamma}
\newunicodechar{Γ}{\Gamma}
\newunicodechar{δ}{\delta}
\newunicodechar{Δ}{\Delta}
\newunicodechar{ε}{\varepsilon}
\newunicodechar{ζ}{\zeta}
\newunicodechar{η}{\eta}
\newunicodechar{θ}{\theta}
\newunicodechar{Θ}{\Theta}
\newunicodechar{ι}{\iota}
\newunicodechar{κ}{\kappa}
\newunicodechar{λ}{\lambda}
\newunicodechar{Λ}{\Lambda}
\newunicodechar{μ}{\mu}
\newunicodechar{ν}{\nu}
\newunicodechar{ξ}{\xi}
\newunicodechar{Ξ}{\Xi}
\newunicodechar{π}{\pi}
\newunicodechar{Π}{\Pi}
\newunicodechar{ρ}{\rho}
\newunicodechar{σ}{\sigma}
\newunicodechar{Σ}{\Sigma}
\newunicodechar{τ}{\tau}
\newunicodechar{υ}{\upsilon}
\newunicodechar{ϒ}{\Upsilon}
\newunicodechar{φ}{\varphi}
\newunicodechar{ϕ}{\phi}
\newunicodechar{Φ}{\Phi}
\newunicodechar{χ}{\chi}
\newunicodechar{ψ}{\psi}
\newunicodechar{Ψ}{\Psi}
\newunicodechar{ω}{\omega}
\newunicodechar{Ω}{\Omega}

\usepackage{xcolor}
\usepackage{xspace}
\usepackage{booktabs}
\usepackage{amsmath,amssymb,mathtools}
\ifdefined\SIAMmacros\else
\usepackage{amsthm}
\fi
\usepackage{dsfont}
\usepackage{bm}
\usepackage{enumerate}
\usepackage{pifont}\def\xmark{\ding{55}} 
\usepackage{tabularx}
\newcolumntype{C}{>{\centering\arraybackslash}X}
\usepackage[commentsnumbered,vlined]{algorithm2e}
\DontPrintSemicolon

\SetCommentSty{commentFont}
\SetKwProg{algo}{Algorithm}{:}{}
\usepackage{etoolbox}
\patchcmd{\SetProgSty}{ArgSty}{ProgSty}{}{}
\SetProgSty{op}
\SetKw{Continue}{continue}
\SetKw{Break}{break}
\SetKw{With}{with}
\SetKw{Retry}{retry}
\usepackage{hyperref}
\usepackage[capitalise,noabbrev]{cleveref}
\makeatletter
\crefname{@theorem}{Theorem}{Theorems}
\makeatother
\crefname{algocf}{Algorithm}{Algorithms}
\crefname{algorithm}{Algorithm}{Algorithms}
\crefname{algocf}{Algorithm}{Algorithms}
\crefname{algocfline}{Line}{Lines}
\usepackage{tikz}
\usetikzlibrary{shapes.geometric,decorations.pathmorphing,arrows.meta,calc,patterns}

\def\itemrefrel#1#2{\stackrel{\text{\lipicsLabel{(\ref{#1})}}}{#2}}
\def\eqrefrel#1#2#3{\stackrel{\text{#2(\ref{#1})}}{#3}}
\def\itemref#1{\lipicsLabel{(\ref{#1})}}
\def\reasonrel#1#2{\stackrel{\text{#1}}{#2}}
\def\myparagraph#1{\smallskip\noindent\textbf{#1}}
\let\paragraph=\myparagraph

\def\e{\mathrm{e}}

\def\cpk{c^{△}_{k}}

\def\pfrac#1#2{\Big(\frac{#1}{#2}\Big)}
\def\op#1{\textsf{#1}}
\def\Sdec{S_{\textrm{dec}}}

\def\QN{Q_{\textrm{next}}}
\def\A{\mathcal{A}}
\def\SA{S_{\A}}
\def\BA{B_{\A}}
\def\P{\mathcal{P}}
\def\e{\mathrm{e}}
\ifdefined\SIAMmacros
	\newcommand\lipicsLabel[1]{{\rmfamily #1}}
\else
	\newcommand\lipicsLabel[1]{\textcolor{darkgray}{\sffamily\upshape\bfseries\mathversion{bold}#1}}
\fi
\newcommand{\OO}{\mathcal{O}}

    \title{Simple Set Sketching}

    \author{Jakob Bæk Tejs Houen}{BARC and University of Copenhagen}{jakn@di.ku.dk}{https://orcid.org/0000-0002-8033-2130}{Part of BARC, supported by the VILLUM Foundation grant 16582}
    \author{Rasmus Pagh}{BARC and University of Copenhagen}{pagh@di.ku.dk}{https://orcid.org/0000-0002-1516-9306}{Part of BARC, supported by the VILLUM Foundation grant 16582}
    \author{Stefan Walzer}{University of Cologne}{walzer@cs.uni-koeln.de}{https://orcid.org/0000-0002-6477-0106}{DFG grant WA 5025/1-1.}
    \authorrunning{J.\ B.\ T.\ Houen, R.\ Pagh, S.\ Walzer}
    \date{}
\fi

\begin{document}

\maketitle

\abstract{
    Imagine handling collisions in a hash table by storing, in each cell, the bit-wise exclusive-or of the set of keys hashing there. This appears to be a terrible idea: For $\alpha n$ keys and $n$ buckets, where $\alpha$ is constant, we expect that a constant fraction of the keys will be unrecoverable due to collisions.

    We show that if this collision resolution strategy is repeated \emph{three times} independently the situation reverses: If $\alpha$ is below a threshold of $≈0.81$ then we can recover the set of all inserted keys in linear time with high probability.

    Even though the description of our data structure is simple, its analysis is nontrivial. Our approach can be seen as a variant of the \emph{Invertible Bloom Filter} (IBF) of Eppstein and Goodrich. 
    While IBFs involve an explicit checksum per bucket to decide whether the bucket stores a single key, we exploit the idea of quotienting, namely that some bits of the key are implicit in the location where it is stored. We let those serve as an implicit checksum. These bits are not quite enough to ensure that no errors occur and the main technical challenge is to show that decoding can recover from these errors.
}

\section{Introduction}




\def\X{\mathcal{X}}
\def\Y{\mathcal{Y}}
Sketching is the idea of representing data in a compact, potentially lossy form. For this introduction imagine that, for some sets $\X$ and $\Y$, a long (typically sparse) sequence $X ∈ \X^u$ is represented via a short sequence $f(X) ∈ \Y^n$ --- the sketch of $X$ --- where $n \ll u$ and $f$ is a (possibly randomized) function.
We speak of \emph{linear sketching} when $(\X,⊕)$ and $(\Y,⊕)$ are groups and $f$ is a linear function, i.e.\ when $f(X⊕X') = f(X)⊕f(X')$ holds (component-wise) for all $X$, $X'$.

Linear sketches of data have appealing properties for applications in streaming or distributed settings~\cite{woodruff:datastreams:2014}.
In particular, such sketches can be merged/updated to form a sketch of the combined data.
This paper considers the case of $\X = \{0,1\}$, meaning the input $X ∈ \{0,1\}^u$ is conceptually a set $S ⊆ [u] := \{1,…,u\}$ of \emph{keys}. We assume that $u+1$ is a power of $2$.

We present a new extremely simple approach for linear sketching of sets. It uses $\Y = \{0,…,u\}$, hence an array $A ∈ \{0,…,u\}^n$ where $n$ is the selected size of the sketch, as well as independent hash functions $h_1,h_2,h_3: [u] \rightarrow [n]$. Given a sketch $A$ of $S ⊆ [u]$ we can add $x ∉ S$ to the sketch (i.e.\ obtain a sketch of $S ∪ \{x\}$) by setting $A[i] ← A[i] ⊕ x$ for $i=h_1(x), h_2(x), h_3(x)$, where $⊕$ denotes bit-wise exclusive-or.

This is indeed a linear sketch if addition in $\X = \{0,1\}$ and $\Y = \{0,…,u\}$ are both understood to be bit-wise exclusive-or.
Merging sketches of sets $S_1$ and $S_2$ will produce a sketch of the symmetric difference $S_1 △ S_2$.
As long as there is only one copy of each element in the sets represented by the sketches we merge, we get a sketch of the union.
We will see that from a sketch of a set $S$ with $n ≥ 1.23\, |S|$ we can recover $S$ with high probability in linear time.

A simple scenario where this is useful is that of \emph{set reconciliation}~\cite{MTZ:setreconciliation:2003}, where two parties, Alice and Bob, have sets $S₁$ and $S₂$ with a large overlap, and want to compute the union $S₁\cup S₂$.
If Alice computes a sketch of $S₁$ and sends it to Bob, he will be able to compute the sketch of $S_1 △ S_2$.
If $n ≥ 1.23\, |S_1 △ S_2|$ then Bob can recover $S_1 △ S_2$ and hence $S₁\cup S₂$ with high probability. Remarkably, the size $n$ of the sketches and hence the amount of information to be transferred is linear in $|S₁ △ S₂|$ rather than being linear in $|S₁ ∪ S₂|$ and therefore close to the information-theoretical lower bound, which holds even if Alice knows which of her elements Bob is missing.

There is a rich literature on streaming algorithms (see e.g.~the surveys~\cite{CJ:samplers:19,McGregor:graphstream:14,woodruff:datastreams:2014}).
Most streaming algorithms are linear sketches over the reals or integers, i.e.\ with $\X = ℤ$ or $\X = ℝ$.
Linear sketches over finite fields like considered in this paper are less well-studied, but are natural in some applications.
For example, consider ``straggler identification''~\cite{EG:StragglerIdentification:2011}, where there is a stream of events of the form $\text{enter}(x)$ and $\text{exit}(x)$, for elements $x\in [u]$ (e.g.\ think about employees entering and leaving a building, or locks being held in a database system).
We want to be able to keep track of which elements have an $\text{enter}(x)$ event without a matching $\text{exit}(x)$ event, assuming that the number of such elements is low (e.g.\ employees left in the building at the end of a working day).
Similarly, for the set reconciliation problem mentioned above, working with a sketch over the field of size two works just as well as working over the integers.

\subsection{Related work}
\label{sec:related}

\begin{figure}
\begin{center}
    \def\features#1{
        \foreach \f in {#1}{
            \makebox[2pt][c]{\textsf{\f}}
        }
    }
    \begin{tabular}{rccccc} 
     \toprule
     method & year & $\frac{\text{space}}{m}$ & $t_{\text{update}}$ & $t_{\text{decode}}$ & techniques\\ [0.5ex]
     \midrule
     randomized $k$-set structure~\cite{Ganguly:k-set:2007} &
        2005 & $O(\log m)$ & $O(\log m)$ & $𝒪(m\log m)$ & \features{A,R,--,--,M} \\ 
     deterministic $k$-set structure~\cite{GM:Deterministic:2008}  &
        2006 & $2$ & $2m$ & $\tilde{𝒪}(m^3)$ & \features{A,--,--,--,M} \\
     symmetric polynomials~\cite{EG:StragglerIdentification:2011} &
        2007 & $1$ & $m$ & $\tilde{𝒪}(m^2)$ & \features{A,--,--,--,--} \\
     IBF~\cite{EG:StragglerIdentification:2011} &
        2007 & $𝒪(\log(m))$ & $𝒪(\log(m))$ & $\tilde{𝒪}(m)$ & \features{--,R,P,C,M} \\
     IBLT ($k=3$)~\cite{GM:IBLT:2011} &
        2011 & 3.666 & 9 & $𝒪(m)$ & \features{--,R,P,C,M} \\
     \midrule
     $\langle$\emph{this paper, $k = 3$}$\rangle$ \phantom{[8]}&
        2023 & 1.222 & 3 & $𝒪(m)$ & \features{--,R,P,--,--} \\
     \bottomrule
    \end{tabular}
    \caption{Comparison of linear sketches for sets and multisets as discussed in \cref{sec:related}, normalised such that decoding is possible if the set size is at most $m$. The space-column counts how many entries have to be stored and $t_{\text{update}}$ counts how many entries are touched by insertions and deletions. The last column indicates which approaches use \textbf{A}lgebraic techniques, \textbf{R}andomisation, \textbf{P}eeling and \textbf{C}hecksums, and which approaches support \textbf{M}ultisets. All randomized sketches have a failure probability of $𝒪(1/m)$. }
    \label{fig:related-work}
\end{center}
\end{figure}

We summarise related work in \cref{fig:related-work}. Each of the listed competitors is a linear set sketch that stores a set $S$ of integers or elements of some finite field. The sketches support insertions and deletions of elements as well as a decode operation that can reproduce $S$ whenever $|S| ≤ m$ for some parameter~$m$. For simplicity we measure required space by counting how many numbers have to be stored, regardless of whether these are from $ℤ$, from a finite field or from $[𝒪(m)]$. Crucially, the space requirement of all sketches only depends on $m$, even though $|S|$ is unlimited. Note that even though decoding is impossible as long as $|S| > m$, it must become possible again if and when $|S| ≤ m$ holds again at some later point.

The following ideas are shared by several of the listed approaches.
\begin{description}
    •[Multisets.] Some approaches allow storing a multiplicity for each element in the set. Unsurprisingly, this tends to double the space requirement.
    •[Algebraic techniques.] A set $S ⊂ ℤ$ of size $m$ is uniquely determined by its power sums $(\sum_{x ∈ S} x^i)_{i ∈ [m]}$. This directly leads to a construction in \cite{EG:StragglerIdentification:2011} using symmetric polynomials and – less directly – to the $k$-set data structures in \cite{GM:Deterministic:2008}. These approaches work deterministically, but have relatively slow update and decode operations.
    •[Randomisation.] A rather primitive linear sketch of a set $S$ of group elements is the sum of the elements. Clearly, when $|S| ≤ 1$ and $S$ does not contain the neutral element then $S$ can be recovered from the sketch.
    
    All randomised approaches use a variant of such a primitive linear sketch in each of a large number of buckets. For each key hash functions select a small number of buckets in which the key is stored. During decoding the hope is that for every key $x$ at least one of its buckets stores no key other than $x$, so that $x$ can be recovered from this bucket.
    
    For randomised approaches decoding may fail even though $|S| ≤ m$. For better comparability we have tuned all competitors to have failure probabilities of $𝒪(1/m)$ in \cref{fig:related-work}.\footnote{If an IBF \cite{EG:StragglerIdentification:2011} is tuned for failure probability $ε$, then space and update times are correspondingly reduced to $𝒪(\log(1/ε))$.}
    
    Two further techniques are often combined with randomised approaches:
    \begin{description}
        •[Checksums.] The decoding algorithm has to decide whether a value $x$ stored in a bucket corresponds to the single key $x$ or to the sum of several keys overlapping in the bucket. Both invertible Bloom filters (IBFs) \cite{EG:StragglerIdentification:2011} and invertible Bloom lookup tables (IBLTs) \cite{GM:IBLT:2011} use explicit hash checksums in each bucket to make this decision. A sanity check proposed in \cite{EG:StragglerIdentification:2011} that is central to our approach is that $x$ can only be stored in a bucket $i$, if $i$ is one of the buckets selected for $x$ by the hash functions. This check can act as an implicit checksum.
        •[Peeling.] Suppose that only a subset $S' ⊂ S$ of the elements in the sketch are directly recoverable due to being alone in a bucket. However, after removing $S'$ from the sketch we obtain a sparser sketch where further elements may be recoverable. Peeling is the natural iterative decoding algorithm arising from the simple insight. It is used by IBFs (though not to its full potential), by IBLTs, and in this paper.

        Our technical contribution is related to the work of Jiang, Mitzenmacher, and Thaler~\cite{JMT:ParallelPeeling:2016}, which studies parallel algorithms for peeling processes such as the one used in IBLTs.
        They show that only $𝒪(\log\log n)$ rounds of peeling are needed in a “breadth-first” peeling approach, similar to the one we use.
    \end{description}
\end{description}

\subsection{Contribution}

We describe the \emph{simple set sketch}, a randomised dynamic set data structure in the spirit of the IBF~\cite{EG:StragglerIdentification:2011} and the IBLT~\cite{GM:IBLT:2011}. At any point in time the sketch represents a set $S ⊆ [u]$ where $u = 2^w-1$, i.e.\ keys are non-zero strings of $w$ bits. Initially $S = ∅$. A \op{toggle} operation can be used to change the membership status of a given key $x ∈ [u]$, meaning that \op{toggle}($x$) changes the represented set from $S$ to $S △ \{x\}$ where $△$ denotes the symmetric difference operator on sets. A \op{merge} operation takes another sketch representing a set $S'$ as input and changes the represented set from $S$ to $S △ S'$.

While no direct membership queries are supported, a \op{decode} operation tries to reconstruct the represented set $S$ in its entirety, and succeeds with high probability under certain conditions discussed below.

The construction uses an array $A = A[1,…,n]$ of $n$ buckets, each of which can store exactly one element of $\{0,…,u\}$, and a constant number $k ≥ 3$ of uniformly random hash functions $h₁,…,h_k : [u] → [n]$. We define $h(x) := \{h₁(x),…,h_k(x)\}$ as a multiset of size exactly $k$, noting that $h(x)$ is an ordinary set with probability $1-𝒪(1/n)$.\footnote{We could have forced $h₁,…,h_k$ to always produce distinct hashes and avoid multisets. However, then $h₁(x),…,h_k(x)$ would not be stochastically independent. So both choices involve mildly annoying (but ultimately inconsequential) technicalities.} The operations are implemented as shown in \cref{fig:algorithms}.

\begin{figure}
    \ifdefined\SIAMmacros\else
    \ \hspace{-2em}
    \fi
    \begin{minipage}{0.5\textwidth}
        \begin{algorithm}[H]
            \algo{initialise}{
                $A[1,…,n] = (0,…,0)$\;
            }
        \end{algorithm}
        \begin{algorithm}[H]
            \algo{toggle($x$)}{
              \For{$i ∈ h(x)$}{
                $A[i] ← A[i] ⊕ x$\;
              }
            }
        \end{algorithm}
        \begin{algorithm}[H]
            \algo{merge($A' ∈ \{0,…,u\}^n$)}{
                \For{$i ∈ [n]$}{
                    $A[i] ← A[i] ⊕ A'[i]$\;
                }
            }
        \end{algorithm}
        \begin{algorithm}[H]
            \algo{looksPure($i ∈ [n]$)}{
                \Return $A[i] ≠ 0 ∧ i ∈ h(A[i])$\;
            }
        \end{algorithm}
    \end{minipage}%
    \ifdefined\SIAMmacros%
        \hspace{-3em}
    \fi%
    \hspace{-3em}%
    \begin{minipage}{0.7\textwidth}
        \begin{algorithm}[H]
            \SetKw{Not}{not}
            \SetKw{ManualIf}{if}
            \algo{decode}{
                $\Sdec ← ∅$\;
                $Q ← \{i ∈ [n] \mid \op{looksPure}(i)\}$\;
                \While{$Q ≠ ∅$}{
                    $\QN ← ∅$\;
                    \For{$i ∈ Q$ \ManualIf \op{\upshape looksPure}$(i)$}{
                        $x ← A[i]$ \tcp{\emph{detected} key $x$} 
                        $\op{toggle}(x)$ \tcp{$S ← S △ \{x\}$}
                        $\Sdec ← \Sdec △ \{x\}$\;  
                        $\QN ← \QN ∪ \{i ∈ h(x) \mid \op{looksPure}(i)\}$\;
                    }
                    $Q ← \QN$
                }
                \If{$A[1,…,n] ≠ (0,…,0)$}{
                    \Return \textsc{failure}\strut
                }
                \Return $\Sdec$ \tcp{correct whp}
            }
        \end{algorithm}
    \end{minipage}
    \caption{Implementation of simple set sketches.}
    \label{fig:algorithms}
\end{figure}

\begin{figure}
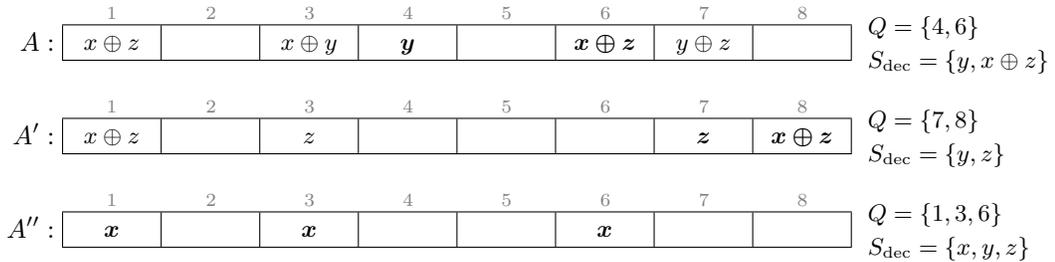

\def\idx#1{\multicolumn{1}{c}{\scriptsize\textcolor{gray}{#1}}}
\begin{align*}
A &: \raisebox{6pt}{\begin{tabularx}{105mm}{|C|C|C|C|C|C|C|C|}
    \idx{1}&\idx{2}&\idx{3}&\idx{4}&\idx{5}&\idx{6}&\idx{7}&\idx{8}\\[-2pt]
    \hline
    $ x⊕z $ & $ $ & $ x⊕y $ & $ \bm{y} $ & $ $ & $ \bm{x⊕z} $ & $ y⊕z $ & $ $ \\    
    \hline
\end{tabularx}}
\begin{tabularx}{20mm}{l}
    $Q=\{4,6\}$\\
    $S_\text{dec}=\{y,x⊕z\}$
\end{tabularx}
\vspace{2mm}\\
A' &: {}\raisebox{6pt}{\begin{tabularx}{105mm}{|C|C|C|C|C|C|C|C|}
    \idx{1}&\idx{2}&\idx{3}&\idx{4}&\idx{5}&\idx{6}&\idx{7}&\idx{8}\\[-2pt]
    \hline
    $ x⊕z $ & $ $ & $ z $ & $  $ & $ $ & $ $ & $ \bm{z} $ & $ \bm{x⊕z} $ \\    
    \hline
\end{tabularx}}
\begin{tabularx}{20mm}{l}
    $Q=\{7,8\}$\\
    $S_\text{dec}=\{y,z\}$
\end{tabularx}
\vspace{2mm}\\
A'' &: {}\raisebox{6pt}{\begin{tabularx}{105mm}{|C|C|C|C|C|C|C|C|}
    \idx{1}&\idx{2}&\idx{3}&\idx{4}&\idx{5}&\idx{6}&\idx{7}&\idx{8}\\[-2pt]
    \hline
    $ \bm{x} $ & $ $ & $ \bm{x} $ & $ $ & $ $ & $ \bm{x} $ & $ $ & $ $ \\    
    \hline
 \end{tabularx}}
 \begin{tabularx}{20mm}{l}
    $Q=\{1,3,6\}$\\
    $S_\text{dec}=\{x,y,z\}$
\end{tabularx}
\end{align*}    
\caption{Example of simple set sketch decoding. The sketch $A$ stores $S = \{x,y,z\}$, assuming $h(x) = \{1,3,6\}, h(y) = \{3,4,7\}, h(z) = \{1,6,7\}$. Moreover, assume $h(x⊕z) = \{3,6,8\}$.
Only $y$ is alone in bucket $4$, but bucket $6$ with the foreign key $x⊕z$ also looks pure, without actually being pure. 
In the first round we therefore have $Q=\{4,6\}$ and toggle $y$ and $x⊕z$, resulting in $A'$.
In the second round buckets $Q=\{7,8\}$ look pure, with $z$ and $x⊕z$, so we toggle these keys and update the set of decoded keys to $S_\text{dec}=\{y,x⊕z\}△\{z,x⊕z\} = \{y,z\}$. 
In the third and final round the remaining key $x$ is recovered from $A''$.}
\label{fig:example}
\end{figure}

The \op{toggle}-operations are commutative and two \op{toggle}$(x)$ operations with the same $x$ cancel. Hence, the state of the data structure is a function of $h$ and the currently stored set $S$. Since $A$ can assume at most $u^n$ states while there are $2^u$ possibilities for $S$, the representation is necessarily “lossy” when $S$ is large and $n$ is small. Like regular IBLTs \cite{GM:IBLT:2011}, \op{decode} relies on \emph{peeling}, meaning we attempt to identify buckets $i$ such that $A[i]$ is the trivial sum of just one key $x$ and hence $A[i] = x$. We call such buckets \emph{pure}. If detected, the key $x$ is toggled – which removes it from the sketch – and $x$ is recorded in the set $\Sdec$ to be returned in the end. A fully successful decode will leave the sketch empty.

To decide whether a bucket $i$ is pure and stores a single key $A[i] = x$ or whether it stores a sum $A[i] = x_1 \oplus \ldots \oplus x_\ell$ of several keys, the \op{looksPure} function checks whether $A[i]$ hashes to $i$, i.e.\ whether $i \in h(A[i])$.
This exploits that when $A[i]$ is single key then $i \in h(A[i])$ always holds, while if $A[i]$ is the sum of several keys then $i \in h(A[i])$ is a coincidence, albeit one that does occur, as we will show later, an expected constant number of times overall.
We leave \op{decode} oblivious of the issue of such \emph{anomalies} and let it trust the output of \op{looksPure}.
That way, it will sometimes erroneously detect a \emph{foreign key}, $z$, that is not actually in the set\footnote{Our notion of a foreign key has nothing to do with the notion of the same name used in the database literature.}.
The algorithm will try to remove~$z$ by calling \op{toggle}$(z)$, but since $z$ is not in the set, this will end up \emph{adding} $z$ to the data structure.
If in the long run the ordinary decoding steps outnumber the anomalous decoding steps, i.e.\ when more keys are removed than added, then $z$ will likely be isolated in a bucket at a later point.
At this point, $z$ will be toggled a second time, this time amounting to an actual removal from the sketch and from~$\Sdec$.
This allows \op{decode} to rectify prior mistakes and return the correct set with high probability.
The implementation uses two nested loops and we call an iteration of the outer loop a \emph{round}. An example for the execution of \op{decode} is given in \cref{fig:example}.
The main technical challenge will be to control the number and properties of anomalous decoding steps so that a successful recovery from the corresponding mistakes occurs with high probability.

In the following theorem the constant $c_k^△$ is known as the \emph{peeling threshold} or the threshold for the occurrence of a $2$-core in a random $k$-uniform hypergraph. The largest, and hence most interesting of these values is $c₃^△ ≈ 0.81$, relevant for $k = 3$ hash functions.
\def\ckp{c_k^{△}}
\begin{theorem}
    \label{thm:IBF-w/o-checksums}
    Assume we have a sketch as explained above with $n$ buckets and $k ≥ 3$ hash functions representing a set $S₀$ of $m$ keys where $\frac{m}{n} < \ckp - ε$ for some $ε > 0$. Then \op{decode} returns $S₀$ in time~$𝒪(n)$ with high probability (whp, meaning with probability $1-\tilde{O}(1/n)$).
\end{theorem}
\noindent
We remark that the error probability $Θ(1/n)$ accounts for three ways in which \op{decode} can fail to return the correct set.
\begin{enumerate}[(1)]
    • \op{decode} may return \textsc{failure}. This is a likely outcome when two keys $x,y ∈ S₀$ satisfy $h(x) = h(y)$, i.e.\ when they share all $3$ hash values. Such keys exist with probability $Θ(1/n)$.
    • \op{decode} may fail to terminate. Assume for instance that $S₀ = \{1,2\}$ with $h(1) = h(2) = \{a,b,c\}$ and $h(3) = \{c,d,e\}$ for some distinct buckets $a,b,c,d,e ∈ [n]$. The algorithm would erroneously select bucket $c$ for decoding since $A[c] = 1⊕2=3$ and $c ∈ h(3)$ – hence \op{looksPure}($c$) is satisfied. This leads to key $3$ being added to the sketch. Afterwards $3$ is correctly detected to be the only key stored in bucket $d$ (or $e$) and toggled a second time, bringing us back to the state we started in.
    A similar constellation of keys exists with probability $Ω(1/n²)$ in any set of $Ω(n)$ keys.
    • \op{decode} may return a set $\Sdec$ with $\Sdec ≠ S₀$. Assume for instance that for $S₀ = \{1,2,3\}$ we have $h(1) = h(2) = h(3)$, which happens with probability $Ω(1/n⁶)$. We then get $\Sdec = ∅$ since the contributions $1⊕2⊕3 = 0$ cancel out everywhere.
\end{enumerate}
The second and third failure cases are more problematic than the first. A practical implementation can prevent (2) by terminating the algorithm with “\textsc{failure}” when it runs unexpectedly long. Moreover, it can reduce the probability of (3) to $2^{-r}$ by introducing a corresponding $r$-bit checksum, i.e.\ maintaining $C = \bigoplus_{x ∈ S} f(x)$ together with the sketch where $f : [u] → \{0,1\}^r$ is a random hash function. Note that $C$ is much more light-weight than the per-bucket checksums used in \cite{EGUV:WhatsTheDifference:2011}.

\subsection{Technical Overview}

From a high level, the analysis has four parts in corresponding subsections.

\paragraph{\ref{sec:def-anomaly} The Issue of Anomalies.}
We connect the “runtime” phenomenon of anomalous decoding steps to the “offline” combinatorial structure of \emph{anomalies}. An anomaly is a set of keys $A = \{x₁,…,x_ℓ\} ⊆ [u]$ with $x₁⊕…⊕x_ℓ = 0$ and a shared bucket $i ∈ h(x₁) ∩ … ∩ h(x_ℓ)$. The presence of any $ℓ-1$ keys from $A$ are, as far as the \emph{centre} bucket $i$ of $A$ is concerned, indistinguishable from the presence of the missing $ℓ$th key from $A$. This may cause $i$ to \op{lookPure}, causing the missing key to be toggled and effectively added to the sketch. Every anomalous decoding step has such an underlying anomaly.

\paragraph{\ref{sec:isolating-anomalies} Isolating Anomalies.}
An anomaly $A$ becomes relevant at runtime, as soon as $|A|-1$ of its keys are present in the sketch. Initially only anomalies with at most one foreign key $x ∈ A \setminus S₀$ are relevant in this sense. We call such anomalies \emph{native anomalies}. However, since native anomalies can cause foreign keys to be added to the sketch, anomalies with two or more foreign keys can become relevant as well, causing additional foreign keys to be added in an escalating cascade.

We show that no such cascade occurs whp. In fact, we show that only $𝒪(1)$ native anomalies exist in expectation (and $𝒪(\log n)$ whp) and that these take only the most harmless of forms with no mutual interaction. Concretely, native anomalies have a “star-shape”, i.e.\ keys share only the centre bucket (formally $|h(A)| = (k-1)|A|+1$) and any two native anomalies have disjoint domains ($h(A₁) ∩ h(A₂) = ∅$).

\paragraph{\ref{sec:working-around-anomalies} Working Around Anomalies.}
Keys that are part of anomalies may be repeatedly toggled by decode, i.e.\ inserted and deleted many times. To obtain a clearer view on the lasting progress that is made, we consider a variant of decode where the dizzying commotion around anomalies is artificially frozen. More precisely, we let $\SA$ be the set of \emph{anomalous keys}, that is, the keys contained in native anomalies, and $\BA = h(\SA)$ the set of \emph{anomalous buckets}.
We then consider a variant \op{decode'} of \op{decode} that is given $\BA$ as an input and is banned from considering these buckets.

With the issue of anomalies out of the picture, \op{decode'} can be analysed with known techniques, which we postpone to \cref{sec:main-lemma}. There we show that all buckets, except for those in $\BA$, are cleared of keys whp. While \op{decode} may (repeatedly) remove and add keys disregarded by \op{decode'}, we show that any key that is removed by \op{decode'} is also permanently removed by \op{decode}. From this we conclude that \op{decode} must reach a state where only anomalous keys are left. It is then not hard to see that these anomalous keys cannot survive in isolation. For each anomaly $A$ and each remaining $x ∈ A$ there are $k-1$ pure buckets only containing $x$, compared to only a single bucket (the centre of $A$) that could look pure without actually being pure. With such a majority of helpful over deceptive information, what is left of the anomaly will unravel within two rounds at most.

\paragraph{\ref{sec:main-lemma} Analysis of \op{decode'}.}
We adapt the analysis of cores in hypergraphs by Molloy~\cite{Molloy05:Cores-in-random-hypergraphs} to our setting with anomalous buckets.
A crucial lemma by Molloy~\cite[Lemma 3]{Molloy05:Cores-in-random-hypergraphs} shows that only a constant fraction of hyperedges remain whp after a constant number of rounds when peeling a fully random $k$-uniform hypergraph.
In our setting, this corresponds to only a constant fraction of the keys remaining after a constant number of iterations of the outer loop of \op{decode} if we have perfect information on which buckets are pure.
We show that since there are only $\OO(\log^2 n)$ anomalous buckets whp, which we block from consideration in \op{decode'}, their effect on the peeling process cannot be too large, and we still obtain that only a constant fraction of the keys remain after a constant number of iterations of the outer loop of \op{decode'}.

We then employ a standard argument to show that if we have fewer than $\delta n$ keys then at least a constant faction of these keys are isolated in a bucket whp.
Now if we have $n'$ isolated keys then we have at least $n' - |B_{\mathcal{A}}|$ buckets that are detected as pure by \op{decode'}.
This shows that \op{decode'} will arrive at a point where at most $\Omega(|\BA|)$ keys from $S \setminus \SA$ are left.
Finally, we need to argue that the last non-anomalous keys are also removed by \op{decode'}, which is done by a technical counting lemma.

\section{Analysis of the Decode Operation}

\subsection{The Issue of Anomalies}
\label{sec:def-anomaly}

We begin by introducing concepts that will come in handy in the subsequent analysis.

We denote by $S₀ ⊆ [u]$ the set stored in the sketch before the \op{decode} operation is executed.
When discussing states of the sketch while \op{decode} is in progress, $S$ refers to the set of keys currently stored in the sketch and $\Sdec$ refers to the current state of the corresponding variable.
Both $S$ and $\Sdec$ may contain native keys, i.e.\ keys from $S₀$, as well as foreign keys, i.e.\ keys from $[u] \setminus S₀$.
Since changes to $S$ and $\Sdec$ happen in sync, $S₀ = \Sdec △ S$ is an invariant of \op{decode}. It implies successful termination if and only if $S = ∅$ is reached.

Each iteration of the while loop carries out a \emph{round} and each iteration of the for-loop where $i$ \emph{looks pure} (i.e.\ \op{looksPure}($i$) holds) carries out a \emph{step} at bucket $i$. We say the key $x = A[i]$ seemingly stored in bucket $i$ is \emph{detected} and toggled. If we in fact had $x ∈ S$, then $S$ loses an element and we speak of a \emph{regular step}, otherwise $S$ gains an element and we speak of an \emph{anomalous step}.

\paragraph{Anomalies.} An anomalous step occurs when bucket $i$ stores several elements $x₁,…,x_{ℓ-1} ∈ [u] \setminus \{x\}$ with $x₁⊕…⊕x_{ℓ-1} = x$. An anomalous step is always linked to an \emph{anomaly} of size $ℓ$.
\begin{definition}
    A set $A = \{x₁,…,x_ℓ\} ⊆ [u]$ with $\bigoplus_{j ∈ [ℓ]} x_j = 0$ and $i ∈ h(x_j)$ for all $j ∈ [ℓ]$ is an anomaly of size $ℓ$ with centre $i ∈ [n]$.\footnote{More precisely: $h(x_j)$ should contain $i$ an odd number of times.}
\end{definition}
An anomaly $A$ of size $ℓ$ is \emph{triggered} if exactly $ℓ-1$ of its keys are stored in the sketch, i.e.\ if $A ∩ S = A \setminus \{x\}$ for some $x ∈ A$, and no other key is stored in the centre bucket $i$. It then appears as though only $x$ is stored in bucket $i$, i.e.\ $i$ looks pure. An anomalous step may then detect key $x$ in $i$ and add $x$ to $S$. Note that $x$ may be native or foreign.

\paragraph{Native anomalies.}
Call an anomaly $A$ a \emph{native anomaly} if it contains at most one foreign key. A native anomaly may already be triggered when decoding starts, or can be triggered simply by removing keys stored in the centre of $A$. In principle, a \emph{foreign anomaly}, i.e.\ an anomaly containing at least two foreign keys, can be triggered, provided that at least one of its keys is added to $S$ during decoding due to different anomalies that are triggered prior to $A$. A non-trivial step in our argument is to show that only native anomalies are triggered whp.

\paragraph{Breadth first decoding.}
It may seem puzzling how decoding could reliably recover from a state where $A ⊆ S$ for some anomaly $A$. Assume the centre of $A$ stores exactly the keys from $A$ and consider the next time a key $x ∈ A$ is removed from $S$. Then $A$ is triggered and it will then appear as though $x = \bigoplus_{x' ∈ A \setminus\{x\}} x'$ is stored in bucket $i$. Since $i$ looks pure, $x$ may be detected at $i$ and hence promptly readded to $S$.
This would indeed be a fatal problem if \op{decode} would maintain the set of buckets to be processed (i.e.\ those that look pure) in a LIFO queue. Instead, \op{decode} proceeds in rounds and a bucket that attains the \op{looksPure} status is only considered in the next round after all buckets that looked pure at the beginning of the round have been processed. Such a “breadth first” way of considering buckets allows for useful work to be done (including the removal of further keys from $A$) before the centre bucket $i$ is considered.

\subsection{Isolating Anomalies}
\label{sec:isolating-anomalies}

Let $\A$ be the set of all native anomalies. In the following we prove that only anomalies from $\A$ are triggered during decoding, that those anomalies have canonical properties and do not interact. This will involve several union bound arguments that are similar to each other in structure. As a warm-up we bound $𝔼[\A]$.


Let us be precise about how a native anomaly arises from the underlying family $(h_j(x))_{x ∈ [u], j ∈ [k]}$ of independent random variables. For any $ℓ ≥ 3$, any set $\{x₁,…,x_{ℓ-1}\} ⊆ S₀$ and any sequence $j₁,…,j_ℓ ∈ [k]$ we call $(\{x₁,…,x_{ℓ-1}\},j₁,…,j_ℓ)$ an \emph{anomaly blueprint}. This blueprint is \emph{realised} if $h_{j₁}(x₁) = … = h_{j_ℓ}(x_ℓ)$ where $x_ℓ := x₁⊕…⊕x_{ℓ-1}$. In that case $A = \{x₁,…,x_ℓ\}$ is a native anomaly. Conversely, every native anomaly realises at least one blueprint (a native anomaly with no foreign key realises at least $ℓ$ blue prints, corresponding to its subsets of size $ℓ-1$). Thus $|\A|$ is at most the number of realised blueprints.
There are $\binom{m}{ℓ-1}k^ℓ$ blueprints with parameter $ℓ$ and each is realised with probability exactly $n^{-ℓ+1}$.
Let $\P$ be the set of all anomaly blueprints and let $E_{P}$ for $P ∈ \P$ be the event that blueprint $P$ is realised. Recall that in the context we are operating we have $c := \frac{m}{n} < \cpk-ε < 1$. We can compute
\begin{align}
    𝔼[|\A|]
    &≤ 𝔼[\,|\{P ∈ \P \mid P \text{ is realised}\}|\,]
     = \sum_{P ∈ \P} \Pr[E_P]
     = \sum_{ℓ ≥ 3} \binom{m}{ℓ-1} k^ℓ · n^{-ℓ+1}\notag\\
    &≤ \sum_{ℓ ≥ 3} \frac{m^{ℓ-1}}{(ℓ-1)!}k^ℓ n^{-ℓ+1}
     = k\sum_{ℓ ≥ 3} \frac{(ck)^{ℓ-1}}{(ℓ-1)!}
     ≤ k\sum_{ℓ ≥ 0} \frac{(ck)^ℓ}{ℓ!} = k\e^{ck} = 𝒪(1). \label{eq:basic-union-bound}
\end{align}
We now show that whp no anomaly $A ∈ \A$ is \lipicsLabel{(ii)} too large, \lipicsLabel{(iii)} contains keys sharing a bucket other than the centre or \lipicsLabel{(iv)} intersects other anomalies in $\A$. We use the notation $h(A) := \bigcup_{x ∈ A} h(x)$.
\begin{lemma}
    \label{lem:canonical-anomalies}
    The following holds whp. 
    \begin{enumerate}[(i)]
        • \label{it:bucket-degree} $∀i ∈ [n]\colon |\{ x ∈ S₀ \mid i ∈ h(x)\}| ≤ \log n$.
        • \label{it:anomaly-size} $∀A ∈ \A\colon |A| ≤ \log n$.
        • \label{it:anomaly-footprint} $∀A ∈ \A\colon |h(A)| = (k-1)|A|+1$.
        • \label{it:anomaly-intersection} $∀A₁ ≠ A₂ ∈ \A\colon h(A₁) ∩ h(A₂) = ∅$.
    \end{enumerate}
\end{lemma}
\begin{proof}\begin{enumerate}[(i)]
    • It is well-known that when $n$ balls are randomly thrown into $n$ bins then the expected maximum load of a bin is $𝒪(\frac{\log n}{\log \log n})$ whp \cite{Gonnet:LongestProbeSequence:81,M:The_Power:1991}, which implies that in our setting every bucket stores $𝒪(\frac{\log n}{\log \log n})$ keys whp. We give a short self-contained proof nonetheless.
    Let $p_{ij}$ be the probability that a specific bin $i ∈ [n]$ stores at least $j ∈ [n]$ keys. A union bound and Stirling's formula gives
    \[ p_{ij} ≤ \binom{km}{j}n^{-j} ≤ \frac{(km)^j}{j!}n^{-j} ≤ \frac{k^j}{j!} ≤ \frac{(k\e)^j}{j^j}.\]
    For $j = 6\frac{\log n}{\log \log n}$ we get for large $n$ and using $x^{\frac{1}{\log x}} = 2$ that
    \[ p_{ij} ≤ \frac{(k^6\e^6)^{\frac{\log n}{\log \log n}}}{(6\frac{\log n}{\log \log n})^{6\frac{\log n}{\log \log n}}} ≤ \frac{2^{\log n}}{(\sqrt{\log n})^{6\frac{\log n}{\log \log n}}} ≤ \frac{n}{2^{3 \log n}} = n^{-2}.\]
    Summing over all $i$ implies that no bin stores $ω(\frac{\log n}{\log \log n})$ keys whp.
    • Since a native anomaly of size $ℓ$ with centre $i$ requires $ℓ-1$ keys from $S₀$ to be stored in $i$, the claim follows from \lipicsLabel{(i)}.
    • Let $A ∈ \A$ be an anomaly and $ℓ = |A|$. There are $kℓ$ relevant hash values. The centre of $A$ occurs as a hash value $ℓ$ times, hence there are at most $(k-1)ℓ+1$ distinct hash values. For there to be at most $(k-1)ℓ$ distinct hash values, an additional identity of two hash values is needed. Since there are at most $\binom{kℓ}{2}$ potential identities that are realised with probability $\frac{1}{n}$ each, we get with calculations similar to~\eqref{eq:basic-union-bound}
    \begin{align*}
        \Pr[∃A ∈ \A: h(\A) ≤ (k-1)|A|]
        &≤ \sum_{ℓ ≥ 3} \binom{m}{ℓ-1} k^ℓ n^{-ℓ+1}\binom{kℓ}{2} \frac{1}{n}\\
        &≤ \sum_{ℓ ≥ 3} \frac{n^{ℓ-1}}{(ℓ-1)!} \frac{k^ℓ}{n^ℓ} k²ℓ² ≤ \frac{1}{n}\sum_{ℓ ≥ 3} \frac{k^{ℓ+2}ℓ²}{(ℓ-1)!} = 𝒪(1/n).
    \end{align*}
    • The main complication stems from the possibility that $A₁$ and $A₂$ may share some keys. We distinguish four cases.
    \begin{description}
        •[Case 1: Shared centres.] Consider the event $E₁$ that there exist $A₁, A₂ ∈ \A$ with $A₁ ≠ A₂$ and the same centre. Assume $|A₁| = ℓ₁$, $|A₂| = ℓ₂$, $|A₁ ∩ A₂| = \bar{ℓ}$ and wlog $A₂ \setminus A₁ ≠ ∅$.
        
        The set $A₂$ can be uniquely identified by $\bar{ℓ}$ keys from $A₁$ and $ℓ₂-\bar{ℓ}-1$ keys from $S₀$. We argue similar to \cref{eq:basic-union-bound}.
        \begin{align*}
        \Pr[E₁] &≤
        \sum_{ℓ₁ ≥ 3} \sum_{ℓ₂ ≥ 3} \sum_{0 ≤ \bar{ℓ} ≤ \min(ℓ₁,ℓ₂-1)} \binom{m}{ℓ₁-1} \binom{ℓ₁}{\bar{ℓ}} \binom{m}{ℓ₂-\bar{ℓ}-1} k^{ℓ₁+ℓ₂-\bar{ℓ}} n^{ℓ₁+ℓ₂-\bar{ℓ}-1}\\
        &
        ≤ \sum_{\bar{ℓ}≥0} \sum_{ℓ₁ ≥ \bar{ℓ}} \sum_{ℓ₂ ≥ \bar{ℓ}+1} \frac{n^{ℓ₁-1}}{(ℓ₁-1)!}\frac{ℓ₁!}{\bar{ℓ}!(ℓ₁-\bar{ℓ})!}\frac{n^{ℓ₂-\bar{ℓ}-1}}{(ℓ₂-\bar{ℓ}-1)!} k^{ℓ₁+ℓ₂-\bar{ℓ}} n^{ℓ₁+ℓ₂-\bar{ℓ}-1}\\
        &≤ \frac{1}{n} \sum_{\bar{ℓ}≥0} \sum_{ℓ₁ ≥ \bar{ℓ}} \sum_{ℓ₂ ≥ \bar{ℓ}+1} \frac{ℓ₁k^{ℓ₁+ℓ₂-\bar{ℓ}}}{\bar{ℓ}!\,(ℓ₁-\bar{ℓ})!\,(ℓ₂-\bar{ℓ}-1)!}\\
        &≤ \frac{1}{n} \sum_{\bar{ℓ} ≥ 0} \frac{k^{\bar{ℓ}}}{\bar{ℓ}!} \sum_{ℓ₁ ≥ \bar{ℓ}} \frac{ℓ₁k^{ℓ₁-\bar{ℓ}}}{(ℓ₁-\bar{ℓ})!} \sum_{ℓ₂ ≥ \bar{ℓ}+1} \frac{k^{ℓ₂-\bar{ℓ}}}{(ℓ₂-\bar{ℓ}-1)!}\\
        &≤ \frac{k\log n}{n} \bigg(\sum_{ℓ ≥ 0} \frac{k^ℓ}{ℓ!}\bigg)^3 = 𝒪\Big(\frac{\log n}{n}\Big) = \tilde{𝒪}(1/n)
        \end{align*}
        where we used $ℓ₁ ≤ \log(n)$ towards the end which we may assume by \itemref{it:anomaly-size}.
        •[Case 2: $\bm{|A₁ ∩ A₂| ≥ 2}$.]
        By \itemref{it:anomaly-footprint} the hashes $h(x)$ and $h(y)$ of two distinct keys $x,y$ in any anomaly $A$ intersect exactly in the centre of $A$ whp. If two anomalies $A₁$ and $A₂$ share two keys $x$ and $y$, they must therefore also share their centre whp. Therefore Case 2 implies Case 1 whp.
        •[Case 3: Distinct centres and $\bm{|A₁ ∩ A₂| = 1}$.] Consider the event $E₃$ that there exist anomalies $A₁$ and $A₂$ with distinct centres and one shared key. Let $ℓ₁ = |A₁|$ and $ℓ₂ = |A₂|$. Now $A₂$ is uniquely identified by one of the $ℓ₁$ keys from $A₁$ and $ℓ₂-2$ keys from $S₀$. We get
        \begin{align*}
            \Pr[E₃] &≤ \sum_{ℓ₁ ≥ 3}\sum_{ℓ₂ ≥ 3}\binom{m}{ℓ₁-1} ℓ₁ \binom{m}{ℓ₂-2} k^{ℓ₁+ℓ₂} n^{-ℓ₁-ℓ₂+2}\\
            &≤ \sum_{ℓ₁ ≥ 3}\sum_{ℓ₂ ≥ 3} \frac{n^{ℓ₁-1}}{(ℓ₁-1)!}ℓ₁\frac{n^{ℓ₂-2}}{(ℓ₂-2)!} k^{ℓ₁+ℓ₂} n^{-ℓ₁-ℓ₂+2}\\
            &≤ \frac{1}{n} \sum_{ℓ₁ ≥ 3}\frac{k^{ℓ₁}ℓ₁}{(ℓ₁-1)!} \sum_{ℓ₂ ≥ 3}\frac{k^{ℓ₂}}{(ℓ₂-2)!}
            ≤ 𝒪(1/n).        
        \end{align*}
        
        •[Case 4: Distinct centres and $\bm{A₁ ∩ A₂ = ∅}$.] Consider the event $E₄$ that there exist anomalies $A₁$ and $A₂$ with distinct centres sharing no key but sharing some $i ∈ h(A₁)∩ h(A₂)$. Let $ℓ₁ = |A₁|$ and $ℓ₂ = |A₂|$ and assume wlog that $i$ is not the centre of $A₁$. One of the $(k-1)ℓ₁$ non-centre hashes of keys in $A₁$ must coincide with one of the $kℓ₂$ hashes from keys in $A₂$. We get 
        \begin{align*}
            \Pr[E₄] &≤ \sum_{ℓ₁ ≥ 3}\sum_{ℓ₂ ≥ 3}\binom{m}{ℓ₁-1} \binom{m}{ℓ₂-1} k^{ℓ₁+ℓ₂}n^{-ℓ₁-ℓ₂+2}(k-1)ℓ₁kℓ₂ \tfrac{1}{n}\\
            &≤ \sum_{ℓ₁ ≥ 3}\sum_{ℓ₂ ≥ 3} \frac{n^{ℓ₁-1}}{(ℓ₁-1)!}\frac{n^{ℓ₂-1}}{(ℓ₂-1)!} k^{ℓ₁+ℓ₂}n^{-ℓ₁-ℓ₂+2}(k-1)ℓ₁kℓ₂ \tfrac{1}{n}\\
            &≤ \frac{1}{n} \sum_{ℓ₁ ≥ 3} \frac{k^{ℓ₁+1}ℓ₁}{(ℓ₁-1)!}\sum_{ℓ₂ ≥ 3} \frac{k^{ℓ₂+1}ℓ₂}{(ℓ₂-1)!} = 𝒪(1/n).\qedhere
        \end{align*}
    \end{description}
\end{enumerate}\end{proof}
We can now derive a concentration bound on the number of native anomalies.
\begin{lemma}
    \label{lem:logn-anomalies}
    There are $|\A| = 𝒪(\log n)$ native anomalies whp.
\end{lemma}
\begin{proof}
    The challenge here is to navigate the fact that anomalies do not occur independently.

    Recall the definition of anomaly blueprints. Let $E_{P,i}$ be the event that a blueprint $P ∈ \P$ is realised at a bucket $i ∈ [n]$. Importantly, $E_{P,i}$ is simply the event that certain random variables in the family $(h_j(x))_{x ∈ [u], j ∈ [k]}$ turn out to be $i$. If we have a sequence $E_{P₁,i₁},…,E_{P_b,i_b}$ of such events pertaining to pairwise distinct buckets $i₁,…,i_b ∈ [n]$ then these events either refer to pairwise distinct random variables and are hence independent, or two events refer to the same random variable and are hence disjoint (i.e.\ inconsistent). Therefore
    \begin{equation}
        \Pr\Big[\bigcap_{j ∈ [b]} E_{P_j,i_j}\Big] ∈ \bigg\{0, \ \ \prod_{j ∈ [b]}{\Pr[E_{P_j,i_j}]} \bigg\} \label{eq:indep-or-incons}
    \end{equation}
    Now define $E_i := \bigcup_{P ∈ \P} E_{P,i}$ to be the event that at least one native anomaly has centre $i$. We can now bound the probability that at least $b$ of these events occur.
    \begin{align*}
        \Pr\big[\sum_{i ∈ [n]}&𝟙_{E_i} ≥ b\big]
        ≤ \sum_{I ⊆ [n], |I| = b} \Pr\Big[\bigcap_{i ∈ I} E_i\Big]
        \reasonrel{sym}{=} \binom{n}{b} \Pr\Big[\bigcap_{i = 1}^b E_i\Big]\\
        &= \binom{n}{b} \Pr\Big[\bigcap_{i = 1}^b \bigcup_{P ∈ \P} E_{P,i}\Big]
         = \binom{n}{b} \Pr\Big[\bigcup_{P₁,…,P_b ∈ \P}\ \bigcap_{i = 1}^b E_{P_i,i}\Big]\\
         &≤ \binom{n}{b} \sum_{P₁,…,P_b ∈ \P} \Pr\Big[ \bigcap_{i = 1}^b E_{P_i,i}\Big]
         \eqrefrel{eq:indep-or-incons}{}{≤} \binom{n}{b} \sum_{P₁,…,P_b ∈ \P} \prod_{i = 1}^b \Pr[E_{P_i,i}]\\
         &= \binom{n}{b} \prod_{i = 1}^b \sum_{P ∈ \P} \Pr[E_{P,i}]
         = \binom{n}{b} \bigg(\sum_{P ∈ \P} \Pr[E_{P,1}]\bigg)^b
         = \binom{n}{b} \bigg(\sum_{ℓ ≥ 3} \binom{m}{ℓ-1} \frac{k^ℓ}{n^{ℓ}} \bigg)^b\\
         &≤ \frac{n^b}{b!} \bigg(\sum_{ℓ ≥ 3} \frac{n^{ℓ-1}}{(ℓ-1)!} \frac{k^ℓ}{n^{ℓ}} \bigg)^b
         ≤ \frac{1}{b!} \bigg(\sum_{ℓ ≥ 3} \frac{k^ℓ}{(ℓ-1)!}\bigg)^b
         = \frac{(k\e^k)^b}{b!} ≤ \frac{(k\e^{k+1})^b}{b^b}.
    \end{align*}
    For $b = Ω(\log n)$ the last term is $𝒪(1/n)$, meaning that only $𝒪(\log n)$ buckets are the centre of native anomalies whp. Since no two native anomalies share a centre whp by \cref{lem:canonical-anomalies} \itemref{it:anomaly-size} this implies that there are $𝒪(\log n)$ native anomalies whp as desired.
\end{proof}
\begin{lemma}
    \label{lem:native-anomalies-only}
    During decoding, only native anomalies are triggered whp.
\end{lemma}
\begin{proof}
    Assume there is a first time $t$ when a foreign anomaly $A₂$ is triggered. Let $i₂$ be its centre and $S$ the set of keys stored in the sketch at time $t$.

    The previously triggered native anomalies may already have introduced some foreign keys to $S$, but by \cref{lem:canonical-anomalies} \itemref{it:anomaly-intersection} these anomalies $A ∈ \A$ have pairwise disjoint domains $h(A)$, so each bucket stores at most one foreign key. The facts that $A₂$ contains at least two foreign keys and that all but one of the keys from~$A₂$ must be present in $S$ in order for $A₂$ to be triggered imply that $A₂$ contains exactly two foreign keys, one of which is present in $S$, call it $y₁ ∈ A₂ ∩ S \setminus S₀$, and one of which is absent, call it $y₂ ∈ A₂ \setminus (S ∪ S₀)$. The presence of $y₁$ in $S$ is due to an anomaly $A₁$ with some centre~$i₁$ that was triggered previously and must be native by choice of $t$. We bound the probability for such a situation to exist, distinguishing two cases. For both we define $ℓ₁ := |A₁|$ and $ℓ₂ := |A₂|$.
    \begin{description}
        •[Case 1: $\bm{i₁ ≠ i₂}$.] We have $A₁ ∩ A₂ = \{y₁\}$ because by \cref{lem:canonical-anomalies} \itemref{it:anomaly-footprint} no two keys from $A₁$ can share $i₂$ as a hash value.
        The pair $(A₁,A₂)$ is uniquely determined by the $ℓ₁-1$ native keys from $A₁$ and the $ℓ₂-2$ native keys from $A₂$. The probability for such a pair to exist is
        \begin{align*}
            &\sum_{ℓ₁ ≥ 3}\sum_{ℓ₂ ≥ 3} \binom{m}{ℓ₁-1}\binom{m}{ℓ₂-2} k^{ℓ₁+ℓ₂} n^{-ℓ₁-ℓ₂+2}\\
            ≤ &\sum_{ℓ₁ ≥ 3}\sum_{ℓ₂ ≥ 3} \frac{n^{ℓ₁-1}}{(ℓ₁-1)!}\frac{n^{ℓ₂-2}}{(ℓ₂-2)!} k^{ℓ₁+ℓ₂} n^{-ℓ₁-ℓ₂+2}
            ≤ \frac{1}{n} \sum_{ℓ₁ ≥ 3} \frac{k^{ℓ₁}}{(ℓ₁-1)!} \sum_{ℓ₂ ≥ 3} \frac{k^{ℓ₂}}{(ℓ₂-2)!} = 𝒪(1/n).
        \end{align*}
        •[Case 2: $\bm{i₁ = i₂}$.] Similar to the proof of \cref{lem:canonical-anomalies} \itemref{it:anomaly-intersection} Case 1, there may now be some number $\bar{ℓ} := |A₁ ∩ A₂ ∩ S₀|$ of shared native keys. Otherwise the computation is similar to Case 1.
        \begin{align*}
            &\sum_{ℓ₁ ≥ 3}\sum_{ℓ₂ ≥ 3} \sum_{\bar{ℓ} ≤ \min\{ℓ₁-1,ℓ₂-2\}} \binom{m}{ℓ₁-1}\binom{ℓ₁-1}{\bar{ℓ}}\binom{m}{ℓ₂-\bar{ℓ}-2} k^{ℓ₁+ℓ₂-\bar{ℓ}-1}n^{-ℓ₁-ℓ₂+\bar{ℓ}+2}\\
            ≤ & \sum_{\bar{ℓ}≥0} \sum_{ℓ₁ ≥ \bar{ℓ}+1} \sum_{ℓ₂ ≥ \bar{ℓ}+2} \frac{n^{ℓ₁-1}}{(ℓ₁-1)!}\frac{(ℓ₁-1)!}{\bar{ℓ}!\,(ℓ₁-\bar{ℓ}-1)!}\frac{n^{ℓ₂-\bar{ℓ}-2}}{(ℓ₂-\bar{ℓ}-2)!} k^{ℓ₁+ℓ₂-\bar{ℓ}-1}n^{-ℓ₁-ℓ₂+\bar{ℓ}+2}\\
            ≤ & \frac{1}{n}\sum_{\bar{ℓ} ≥ 0} \frac{k^{\bar{ℓ}}}{\bar{ℓ}!} \sum_{ℓ₁ ≥ \bar{ℓ}-1} \frac{k^{ℓ₁-\bar{ℓ}-1}}{(ℓ₁-\bar{ℓ}-1)!} \sum_{ℓ₂ ≥ \bar{ℓ}-2} \frac{k^{ℓ₂-\bar{ℓ}}}{(ℓ₂-\bar{ℓ}-2)!}
            ≤ 𝒪(1/n).
        \end{align*}
        Taken together, no such situation arises whp.\qedhere
    \end{description}
\end{proof}
\begin{lemma}
    \label{lem:no-surprise-keys}
    Let $\SA := \bigcup_{A ∈ \A} A$ be the set of \emph{anomalous keys}. During decoding, we have $S ⊆ S₀ ∪ \SA$ at all times whp.
\end{lemma}
\begin{proof}
    This follows from induction. Initially we have $S = S₀$. Any regular decoding step removes an element from $S$. Any anomalous decoding step adds a key $y ∈ A$ for some anomaly $A$ that has been triggered. By \cref{lem:native-anomalies-only} we have $A ∈ \A$ and hence $y ∈ \SA$, maintaining the invariant.
\end{proof}

\subsection{Working around anomalies}
\label{sec:working-around-anomalies}

Let $\SA := \bigcup_{A ∈ \A} A$ be the set of \emph{anomalous keys} and $\BA = h(\SA)$ the set of \emph{anomalous buckets}.
Consider a variant \op{decode'} of \op{decode} (see \cref{fig:algorithms}) that receives the set $\BA$ of anomalous buckets as a parameter and ignores these buckets, say by pretending that no $i ∈ \BA$ ever satisfies $\op{looksPure}(i)$. Similar to $S$, we use $S'$ to track the set of keys stored in the sketch over time when \op{decode'} is used. No anomalous decoding steps can occur in \op{decode'}, because the first anomalous decoding step would have to be at the centre of a native anomaly, but these centres are contained in $\BA$ and banned from consideration. In particular, elements are only ever removed from $S'$, never added.

Recall that by a \emph{round} of \op{decode} we mean one iteration of the while-loop. Rounds typically comprise many decoding steps. 
To ensure the $r$th round is well-defined for each $r ∈ ℕ$ we imagine that if and when the algorithm terminates (because both $Q$ and $Q_{\textrm{next}}$ are empty) an infinite number of further rounds take place that contain no steps.

We now show that \op{decode} correctly identifies at least as many keys from $S₀$ as \op{decode'}. For this we define~$S_r$ for $r ∈ ℕ₀$ to be the set of keys stored in the sketch at the start of round $r+1$ when \op{decode} is used and $S_r'$ to be the corresponding set when \op{decode'} is used.
\begin{lemma}
    \label{lem:decode-vs-decode-prime}
     We have $S_r ∩ S₀ ⊆ S_r'$ for all $r ∈ ℕ₀$ whp.
\end{lemma}

\begin{proof}
    We proceed by induction. At the start of round $1$ we have $S₀ = S₀'$ so there is nothing to show. Now assume that at the start of some round $r$ we have $S_r ∩ S₀ ⊆ S_r'$. To show $S_{r+1} ∩ S₀ ⊆ S_{r+1}'$ we consider different cases for $x ∈ S_{r+1}∩S₀$. We may assume that the high-probability guarantees from \cref{lem:native-anomalies-only,lem:no-surprise-keys} hold.
    \begin{description}
        •[Case 1: $\bm{x} ∈ A$ for some $A ∈ \A$.] We have $h(x) ⊆ \BA$, i.e.\ the buckets of $x$ are banned from consideration in \op{decode'}. Since $x ∈ S₀ = S₀'$ and $x$ can never be toggled in \op{decode'} we have $x ∈ S_{r+1}'$ as well.
        •[Case 2: $\bm{x}$ is not part of a native anomaly.] Combined with \cref{lem:native-anomalies-only}, $x$ is not contained in any anomaly that is triggered and could not have been readded to $S$. It was therefore already in $S$ when round $r$ started, meaning $x ∈ S_{r} ∩ S₀$. The induction hypothesis gives $x ∈ S_{r}'$. We have to show $x ∈ S_{r+1}'$. Assume for contradiction that $x ∉ S_{r+1}'$. Then $x$ was removed in round $r$ of \op{decode'}. Thus one of its buckets $b ∈ h(x)$ was in $Q$ at the start of round $r$ of \op{decode'}. Hence $b ∉\BA$ and we had \op{looksPure}($b$) at some prior time. No anomaly can be triggered at $b$ by \cref{lem:native-anomalies-only} and \op{looksPure($b$)} really means that only one key is stored in $b$. Hence $x$ is the only key in $S_r'$ with $b$ as a hash. By induction hypothesis, $x$ is the only key in $S_r ∩ S₀$ with $b$ as a hash. Moreover, since $S_r \setminus S₀ ⊆ \SA$ by \cref{lem:no-surprise-keys} we have $h(S_r \setminus S₀) ⊆ \BA \not\owns b$ so $x$ is the only key in $S_r$ with $b$ as a hash and \op{looksPure($b$)} holds at the start of round $r+1$ of \op{decode}. This implies that $b$ is in $Q$ at the start of round $r$ of \op{decode}. Again using that no triggered anomaly can add keys to bucket $b$, we conclude that $x$ is detected and removed during round $r+1$ of \op{decode}. Moreoever, $x$ cannot be readded afterwards since $x ∉ \SA$. This implies $x ∉ S_{r+1}$, contradicting the choice of $x ∈ S_{r+1} ∩ S₀$. Since the assumption $x ∉ S_{r+1}'$ led to this contradiction we have $x ∈ S_{r+1}'$ as desired. \qedhere
    \end{description}
\end{proof}

\noindent
On the other hand we will show in \cref{sec:main-lemma} that \op{decode'} succeeds in decoding everything except keys in $\SA$ and in fact does so in a polylogarithmic number of rounds:
\begin{lemma}
    \label{lem:peel-all-but-anomalies}
    With high probability, \op{decode'} achieves $S_R' ⊆ \SA$ for some $R = \tilde{𝒪}(1)$.
\end{lemma}
%
%
Before showing how this implies our main theorem we deal with a technicality concerning the implementation of $Q$ and $\QN$. As the names suggest, we have queues in mind, such as a LIFO or FIFO queue. However, since we cannot afford to check if an element is already in $\QN$ whenever we are about to add something to $\QN$ this effectively implements $Q$ and $\QN$ as multisets. A duplicated bucket $i$ in $\QN$ means a duplicated execution of the for-loop for bucket $i$ in the next round. None of the previous arguments hinge on this, but one might worry that with excessive duplication the running time gets out of hand. We are reluctant to resolve the issue by using a set data structure for $\QN$, because this would compromise the simplicity of \op{decode}. Moreover the issue can be resolved with the following simple Lemma.
\begin{lemma}
    \label{lem:copies-in-Q}
    Assume an implementation of \op{decode} realises $Q$ and $\QN$ as multisets, e.g.\ using FIFO queues.
    Then whp the following is true for all $A ∈ \A$. Together $Q$ and $\QN$ never contain more than two copies of the centre $i$ of $A$. If they contain two copies of $i$, then $i$ stores at most one key.
\end{lemma}

\begin{proof}
    Since no anomaly other than $A$ affects $i$ by \cref{lem:canonical-anomalies} \itemref{it:anomaly-intersection} and \cref{lem:native-anomalies-only}, there are only two reasons for adding $i$ to $\QN$:
    \begin{enumerate}[(i)]
        • The anomaly $A$ is triggered, meaning the state of $i$ changed such that $i$ now stores $|A|-1$ keys from $A$ and no other key, or
        • $i$ stores only a single key.
    \end{enumerate}
    Reason \lipicsLabel{(i)} may occur several times but the necessary state change in between two occurrences must include the addition of a key and the removal of a key. A key can only be added to $i$ due to an anomalous decoding step at $i$, which consumes a copy of $i$ from $Q$, maintaining the invariant. While reason \lipicsLabel{(ii)} may push a second copy of $i$ into $\QN$, this can only happen once since no anomalous decoding steps can add keys to $i$ afterwards (recall that $|A| ≥ 3$).
\end{proof}

\begin{proof}[Proof of \cref{thm:IBF-w/o-checksums}]
    We may assume that the high probability events from all previous lemmas hold. Let $R = \tilde{𝒪}(1)$ be the number of rounds from \cref{lem:peel-all-but-anomalies} needed until $S_R' ⊆ \SA$. Since $S_R ∩ S₀ ⊆ S_R'$ by \cref{lem:decode-vs-decode-prime} and $S_R \setminus S₀ ⊆ \SA$ by \cref{lem:no-surprise-keys} we have $S_R ⊆ \SA$ as well, i.e.\ only anomalous keys might remain after $R$ rounds of \op{decode}. We now show that two more rounds suffice (i.e.\ $S_{R+2} = ∅$) by showing $S_{R+2} ∩ A = ∅$ for any $A ∈ \A$.

    Since native anomalies have non-overlapping domain $h(A)$ we may consider each $A$ in isolation. Let $i$ be the centre of $A$. 
    Consider the beginning of round $R+1$ (when $\QN = ∅$).
    If $Q$ contains two copies of $i$, then by \cref{lem:copies-in-Q} we have $|S_R∩A| = 1$ and this single key is clearly removed in the next round. Otherwise \cref{lem:copies-in-Q} guarantees that there is at most one copy of $i$ in $Q$. Each $x ∈ A$ is the only key stored in the $k-1$ buckets $h(x) \setminus \{i\}$ by \cref{lem:canonical-anomalies} \itemref{it:anomaly-footprint}. These buckets “\op{lookPure}” and are hence all contained in $Q$. Therefore, every $x ∈ A$ is removed within round $R+1$ (at least once). When bucket $i$ is processed, at most one key from $A$ is added. This leaves us with at most one key from $A$ after $R+1$ rounds and hence no key from $A$ after $R+2$ rounds.
    This concludes the proof that $S = ∅$ after $\tilde{𝒪}(1)$ rounds of \op{decode} and hence that \op{decode} terminates with $\Sdec = S₀$ whp.

    A final issue is the running time. We assume $Q$ and $\QN$ are implemented as queues. Let $a$ be the total number of anomalous decoding steps. By \cref{lem:logn-anomalies} there are $\tilde{𝒪}(1)$ native anomalies whp, by \cref{lem:native-anomalies-only} no other anomalies are ever triggered whp and by \cref{lem:copies-in-Q} each anomaly can lead to at most one anomalous decoding step per round whp. Hence $a = \tilde{𝒪}(1)$ whp. Since regular decoding steps remove a key and anomalous steps add a key, there are $m+a$ regular steps, giving $m + 2a$ decoding steps in total whp. The total number of entries added to $Q$ and $\QN$ is then at most $n+(k-1)(m+2a) = 𝒪(n)$, which accounts for $n$ additions before the while-loop and $k-1$ additions per decoding step. Since every iteration of the for-loop consumes an element from $Q$, there are $𝒪(n)$ for-loop iterations whp.
\end{proof}

\subsection{Analysis of \op{decode'}}
\label{sec:main-lemma}

The arguments used in the following proof of \cref{lem:peel-all-but-anomalies} should not be regarded as completely novel. The fact that most keys can be removed is closely related to the analysis of cores in hypergraphs as discussed by Molloy \cite{Molloy05:Cores-in-random-hypergraphs} and the required number of rounds of peeling has been studied in a similar case in more detail by Jiang, Mitzenmacher and Thaler \cite{JMT:ParallelPeeling:2016} who prove that $Θ(\log \log n)$ rounds are necessary and sufficient. We adapt these existing works to our setting with anomalies.

We recall some facts about hypergraph peeling closely related to our setting. The hypergraph to consider here is $H = ([n],\{h(x) \mid x ∈ S₀\})$. To avoid parallel terminologies, we continue to call $i ∈ [n]$ a bucket (rather than a vertex) and speak of a key $x$ (effectively referring to the hyperedge $h(x)$). We do however adopt graph theoretic notions such as the \emph{incidence} of a bucket $i$ to a key $x$ (meaning $i ∈ h(x)$), the degree of a bucket (its number of incidences) or the $r$-neighbourhood of a bucket or key (the set of buckets and keys reachable by traversing at most $r$ hyperedges).

The peeling process on $H$ proceeds in rounds. In every round the set of buckets $B₁ ⊆ [n]$ of degree $1$ is determined. Then all keys incident to a bucket from $B₁$ are removed. This may cause further buckets to lose incidences, creating new buckets of degree $1$, which are then handled in the next round. If this process eventually removes all keys, then the original $H$ is called \emph{peelable}. Peelability is for instance exploited to decode IBLTs \cite{GM:IBLT:2011}, construct error correcting codes \cite{LMSS:Efficient_Erasure:2001} and to solve random linear systems to construct perfect hash functions \cite{BPZ:Practical:2013} or Bloom filter alternatives~\cite{GL:XorFilters:2020}.

A density threshold $\cpk$ for peelability is known, meaning a fully random $k$-uniform hypergraph (like $H$ above) is peelable whp if $\frac{m}{n} < \cpk-ε$ and not peelable whp if $\frac{m}{n} > \cpk + ε$ \cite{Molloy05:Cores-in-random-hypergraphs}. Moreover, the following lemma guarantees that below the threshold a constant number of rounds suffice to remove most keys whp:
\begin{lemma}[Molloy {\cite[Lemma 3]{Molloy05:Cores-in-random-hypergraphs}}]
    \label{lem:molloy}
    For any $ε,δ > 0$, there exists $R ∈ ℕ$ such that after peeling a $k$-uniform hypergraph with hyperedge density $\frac{m}{n} < \cpk-ε$ for $R$ rounds at most $δn$ hyperedges remain whp.\footnote{Strictly speaking, Molloy's Lemma only claims a probability of $1-o(1)$. However, the tool utilised in his proof (Azuma's inequality in his Lemma 7) is strong enough to support our “\emph{whp}”.}
\end{lemma}
Our algorithm \op{decode'} behaves almost exactly like a peeling algorithm.
The only substantial difference is that the buckets from $\BA$ are never considered regardless of their degree. As it turns out, this disturbance is too weak to affect the guarantee given in \cref{lem:molloy}, as we show now.

\begin{lemma}
    \label{lem:early-rounds}
    For any $δ > 0$, there exists $R ∈ ℕ$ such that $|h(S_{R}')| ≤ δn$ whp.
\end{lemma}
\begin{proof}
    The density condition $\frac{m}{n} < \ckp-ε$ is part of the requirement of \cref{thm:IBF-w/o-checksums} – the context in which we are operating. Without the complication of anomalous buckets we could obtain a constant $R$ such that $|S_{R}'| ≤ \frac{δ}{2k}n$ whp by \cref{lem:molloy}. Since peeling is a local algorithm, a key $x$ is only affected by the restriction regarding $\BA$ if there is some $i ∈ \BA$ within the $R$-neighbourhood of $x$. Since the maximum degree of any bucket is at most $\log n$ whp by \cref{lem:canonical-anomalies} and because $|\BA| ≤ k|\SA| ≤ k \log(n)|\A| = 𝒪(\log²(n))$ by \cref{lem:logn-anomalies} there are whp at most $𝒪(\log^{2+R}(n))$ buckets in the $R$-neighbourhoods of buckets in $\BA$ that could be affected in this way. We obtain $|S_{R}'| ≤ \frac{δ}{2k}n + 𝒪(\log^{2+R}(n)) ≤ \frac{δ}{k}n$ whp and therefore $|h(S_{R}')| ≤ k|S_{R}'| ≤ δn$ whp as desired.
\end{proof}
For a set $I ⊆ [n]$ of buckets let $S_I := \{ x ∈ S₀ \mid h(x) ⊆ I\}$ be the set of keys \emph{induced} by $I$.
\begin{lemma}
    \label{lem:intermediate-rounds}
    There exist constants $N∈ℕ$ and $δ,γ > 0$ such that, whp, any set $I ⊆ [n]$ of buckets with $N ≤ |I| ≤ δn$ satisfies $|S_I| ≤ (2-γ)|I|/k$.
\end{lemma}
\begin{proof}
    We start by bounding the probability $p_i$ that there exists a set $I$ of size $i := |I|$ with $|S_I| ≥ s$ where $s := s(i) := ⌈(2-γ)i/k⌉$, using a union bound. At the line break we use that $i = Θ(s)$ and that hence $\e^{i+s}(\frac{s}{i})^s ≤ C^i$ for a suitable constant $C$. We also use $k ≥ 3$ and choose $γ = \frac{1}{4}$.
    \begin{align*}
        p_i &≤ \binom{n}{i} \binom{m}{s} \pfrac{i}{n}^{ks}
        ≤ \pfrac{ne}{i}^{i} \pfrac{ne}{s}^s \pfrac{i}{n}^{ks}
        = \e^{i+s}\pfrac{i}{s}^s · \pfrac{n}{i}^i \pfrac{n}{i}^s \pfrac{i}{n}^{ks}\\
        &≤ C^{i} \pfrac{n}{i}^i \pfrac{i}{n}^{(k-1)(2-γ)i/k}
        ≤ C^{i} \pfrac{i}{n}^{-i+\frac 23(2-γ)i}
        ≤ C^{i} \pfrac{i}{n}^{(\frac 13 - \frac 23γ)i}
        ≤ C^{i} \pfrac{i}{n}^{i/6}
        = \pfrac{iC^6}{n}^{i/6}.
    \end{align*}
    We fix $N := 18$ and $δ := C^{-6}/2$. This allows for the following union bound on the probability that a set $I$ of some size $N ≤ |I| ≤ δn$ induces $s$ or more keys.
    \begin{align}
        \sum_{i = N}^{δn} p_i
        &≤ \sum_{i = N}^{δn} \pfrac{iC^6}{n}^{i/6}
        = \sum_{i = N}^{\sqrt{n}}    \pfrac{iC^6}{n}^{i/6}
        + \sum_{i = \sqrt{n}+1}^{δn} \pfrac{iC^6}{n}^{i/6}\label{eq:simple-union-bound-final-step}\\
        &≤ \sum_{i = N}^{\sqrt{n}}    \pfrac{\sqrt{n}C^6}{n}^{N/6}
        +  \sum_{i = \sqrt{n}+1}^{δn} \pfrac{δnC^6}{n}^{\sqrt{n}/6}
        ≤ \sqrt{n}·\pfrac{C^6}{\sqrt{n}}^{3}
        +  n·\pfrac{1}{2}^{\sqrt{n}/6} = 𝒪(1/n).\qedhere\notag
    \end{align}
\end{proof}

\begin{lemma}
    \label{lem:late-rounds}
    There exists a constant $δ > 0$ such that the following holds whp. For any set of keys $S^* ⊆ S₀ \setminus \SA$ and $I := h(S^*)$, $i := |I|$, $i ≤ δn$, $s := |S^*|$ we have $s < (2i-|I ∩ \BA|)/k$.
\end{lemma}

\begin{proof}
    We use a union bound from the perspective of the set $I$. Concretely we bound for fixed $i,a ∈ ℕ$ the probability $p_{i,a}$ that there exists a set $I ⊆ [n]$ of size $i$ that satisfies $|I ∩ \BA| ≥ a$ as well as $|S_I \setminus \SA| ≥ s$ for $s := s(i,a) := \max(\frac{i}{k},⌈(2i-a)/k⌉)$, i.e.\ $I$ contains at least $a$ anomalous buckets and induces at least $s$ keys from $S₀ \setminus \SA$. It then suffices to show that the sum over all $p_{i,a}$ is $𝒪(1/n)$. Note that we may assume $s ≥ \frac{i}{k}$ because of “$I = h(S^*)$”. This ensures $s = Θ(i)$.

    We enumerate all ways in which $I$ might arise (though including some inconsistent combinations in our counting), thereby explaining \cref{eq:crazy-union-bound} below, from left to right.
    \begin{itemize}
        • There are $\binom{n}{i}$ ways to select $I$,
        • and $\binom{m}{s}$ ways to select a set of $s$ keys from $S₀ \setminus \SA$ to be induced by $I$. 
        • We specify a number $d ≤ a$ of disjoint native anomalies that are to contribute to $I ∩ \BA$ (note that by \cref{lem:canonical-anomalies} \itemref{it:anomaly-intersection} we need not worry about the possibility of intersecting native anomalies).
        • For each $j ∈ [d]$ we specify properties of the $j$th anomaly $A_j$, namely:
        \begin{itemize}
            • A number $a_j ≥ 1$ of elements from $I ∩ \BA$ that $A_j$ accounts for. These numbers should satisfy $\sum_j a_j = a$. Each $a_j$ is a lower bound on $|I ∩ h(A_j)|$.
            • The size $ℓ_j := |A_j|$ of $A_j$.
            • A set of $ℓ_j-1$ keys from $S₀$ that uniquely determine $A_j$ itself.
            • For each $ x ∈ A_j$, which of its $k$ hashes should point to the anomaly's centre.
            • And finally, which subset of the $k$ hashes of $x ∈ A_j$ must fall within $I$.
        \end{itemize}
    \end{itemize}
    The probability for such a precisely specified constellation is either $0$ (if the specification is inconsistent) or has a simple form shown in the following formula. The term $(\frac{i}{n})^{ks}$ accounts for the selected keys being induced by $I$, $n^{-ℓ_j+1}$ accounts for the keys in $A_j$ forming an anomaly and $(\frac{i}{n})^{a_j}$ accounts for the selected hashes from keys in $A_j$ actually falling into $I$. It should be clear that these three probabilities are independent. We obtain the following.
    \begin{align}
        \label{eq:crazy-union-bound}
        p_{i,a} &≤ 
        \underbrace{\tbinom{n}{i} \tbinom{m}{s} \sum_{d ≤ a}\,\, \sum_{a₁ + … + a_d = a}\,\,
        \sum_{ℓ₁ ≥ 3} \tbinom{m}{ℓ₁-1} (k2^k)^{ℓ₁} … \sum_{ℓ_d ≥ 3} \tbinom{m}{ℓ_d-1}
        (k2^k)^{ℓ_d}}_{\text{sum over all events contributing to the union bound}}
        ·\underbrace{
            (\tfrac{i}{n})^{ks} \prod_{j ∈ [d]} n^{-ℓ_j+1} (\tfrac{i}{n})^{a_j}
        }_{\text{probability of the event}}\\
        &≤\tbinom{n}{i} \tbinom{m}{s} (\tfrac{i}{n})^{ks} \sum_{d ≤ a}\,\, \sum_{a₁ + … + a_d = a}\,\,
        \prod_{j ∈ [d]} \underbrace{\sum_{ℓ_j ≥ 3}\tbinom{m}{ℓ_j-1} (k2^k)^{ℓ_j} n^{-ℓ_j+1} (\tfrac{i}{n})^{a_j}}_{(*)}\notag
    \end{align}
    Let us continue to simplify $(*)$, using the constant $C := k2^k\e^{k2^k}$.
    \begin{align*}
        (*) &≤
        \sum_{ℓ_j ≥ 3}\frac{n^{ℓ_j-1}}{(ℓ_j-1)!} (k2^k)^{ℓ_j} n^{-ℓ_j+1} (\tfrac{i}{n})^{a_j}
        ≤ \pfrac{i}{n}^{a_j} \sum_{ℓ_j ≥ 3}\frac{(k2^k)^{ℓ_j}}{(ℓ_j-1)!} ≤ C \pfrac{i}{n}^{a_j}.
    \end{align*}
    We continue the computation in \cref{eq:crazy-union-bound} using that any $a ≥ 1$ can be written as the sum of a sequence of positive integers in precisely $2^{a-1}$ ways, and $a = 0$ can be written in precisely one way as the empty sum.
    \begin{align*}
        p_{i,a} &≤ \tbinom{n}{i} \tbinom{m}{s} (\tfrac{i}{n})^{ks} \sum_{d ≤ a}\,\, \sum_{a₁ + … + a_d = a}\,\,
        \prod_{j ∈ [d]} C \pfrac{i}{n}^{a_j}
        ≤ \tbinom{n}{i} \tbinom{m}{s} (\tfrac{i}{n})^{ks} \sum_{d ≤ a}\,\, \sum_{a₁ + … + a_d = a}\,\,
        C^a \pfrac{i}{n}^{a}\\
        &≤ \tbinom{n}{i} \tbinom{m}{s} (\tfrac{i}{n})^{ks} 2^{a}
        C^a \pfrac{i}{n}^{a}
        ≤ \pfrac{\e n}{i}^i \pfrac{\e n}{s}^s \pfrac{i}{n}^{ks} 2^{a}
        C^a \pfrac{i}{n}^{a}\\
        &≤ \e^{i+s} \pfrac{i}{s}^s 2^a C^a · \pfrac{n}{i}^i\pfrac{n}{i}^s \pfrac{i}{n}^{ks} \pfrac{i}{n}^a ≤ (C')^{i} · \pfrac{i}{n}^{(k-1)s+a-i}
    \end{align*}
    where $C'$ is another constant (here we use $a ≤ i$ and $s = Θ(i)$). We bound the exponent using $s ≥ (2i-a)/k$ (and hence $ks+a ≥ 2i$) as well as $k ≥ 3$ as follows.
    \[\ (k-1)s+a-i = \tfrac{k-1}{k}(ks + a) + a/k - i ≥ \tfrac{k-1}{k}·2i + a/k - i = \tfrac{k-2}{k}i + a/k ≥ i/3+a/k.\]
    Considering that the exponent was an integer we finally obtain
    \[ p_{i,a} ≤ (C')^{i} · \pfrac{i}{n}^{⌈i/3+a/k⌉} =: q_{i,a}.\]
    To conclude the prove we need to bound $P := \sum_{1 ≤ i ≤ δn}\sum_{a ≥ 0} p_{i,a}$ for some $δ > 0$ of our choosing. The sum over $a$ is effectively a geometric sum with
    \[ P ≤ \sum_{1 ≤ i ≤ δn} \sum_{a ≥ 0} q_{i,a} ≤ \sum_{1 ≤ i ≤ δn} k·q_{i,0}\sum_{a ≥ 0} \pfrac{i}{n}^a ≤ k·\sum_{1 ≤ i ≤ δn} q_{i,0}·\frac{1}{1-i/n} ≤ \frac{k}{1-δ}·\sum_{1 ≤ i ≤ δn} q_{i,0}. \]
    To bound the sum over the $q_{i,0}$ we have to make use of the “$⌈·⌉$” for the leading terms, e.g.\ like this:
    \[ \sum_{1 ≤ i ≤ δn} q_{i,0} ≤ 𝒪(1/n) + \sum_{9 ≤ i ≤ δn} \pfrac{i(C')^3}{n}^{i/3}.\]
    To bound the remaining sum the same idea as in \cref{eq:simple-union-bound-final-step} works.
\end{proof}

\begin{proof}[Proof of \cref{lem:peel-all-but-anomalies}.]
    We assume that the high probability guarantees from \cref{lem:early-rounds,lem:intermediate-rounds,lem:late-rounds} hold. Let $N$ and $γ$ be the constants from \cref{lem:intermediate-rounds}. Moreover, \cref{lem:intermediate-rounds,lem:late-rounds} each guarantee the existence of a constant named “$δ$”. Let $δ$ be the smaller of the two and apply \cref{lem:early-rounds} for this $δ$. This yields another constant $R = R(δ)$ such that $|h(S_R')| ≤ δn$. We refer to the first $R$ rounds as the \emph{early rounds} of peeling.

    Now consider any round $r ≥ R$ such that $S_r' \setminus \SA ≠ ∅$. We define $I := h(S_r' \setminus \SA)$ and $i := |I|$. Since $S_r' ⊆ S_R'$ we have $i ≤ |h(S_R')| ≤ δn$ so we may apply both \cref{lem:intermediate-rounds,lem:late-rounds} to $I$ (see below). It is useful to distinguish three kinds of buckets in $I$.
    \def\IA{I_{\A}}
    \def\IO{I₁}
    \def\IT{I_{2+}}
    \def\iA{i_{\A}}
    \def\iO{i₁}
    \def\iT{i_{2+}}
    \begin{itemize}
        • $\IA := I ∩ \BA$: anomalous buckets in $I$.
        • $\IO := \{b ∈ I \mid ∃!(x,j) ∈ S_I × [k]\colon h_j(x) = b\}$ where “$∃!$” means “there exists exactly one”. These are buckets not in $\IA$ that have degree $1$ with respect to $S_I$.
        • $\IT := I \setminus (\IA ∪ \IO)$: Buckets not in $\IA$ that have degree $2$ or more with respect to $S_I$.
    \end{itemize}
    Denote the numbers of these buckets with $\iA, \iO, \iT$, respectively. We have
    \begin{enumerate}[(i)]
        • \label{eqit:def-i}
        $i = \iO + \iT + \iA$ by definition.
        • \label{eqit:degree}
        $k|S_I| ≥ k|S_I \setminus \SA| ≥ i + \iT$ since each $b ∈ I$ is hashed to at least once and each $b ∈ \IT$ is hashed to at least twice, both by keys from $S_I \setminus \SA$.
        • \label{eqit:interm}
        $|S_I| ≤ (2-γ)i/k$ and equivalently $2i ≥ \frac{2}{2-γ}k|S_I|$ by \cref{lem:intermediate-rounds} if $i ≥ N$.
        • \label{eqit:late}
        $|S_I \setminus \SA| < (2i-\iA)/k$ by \cref{lem:late-rounds}.
    \end{enumerate}
    We distinguish two further types of rounds depending on $i$. We show $i₁ > 0$ in both cases.
    \begin{description}
        •[Intermediate rounds with $i = ω(\log² n)$.] We compute
        \begin{align*}
            \iO &\itemrefrel{eqit:def-i}{=} i - \iT - \iA
            = 2i - i - \iT - \iA
            \itemrefrel{eqit:interm}{≥} \tfrac{2}{2-γ}k|S_I| - i - \iT - \iA\\
            &\itemrefrel{eqit:degree}{≥} \tfrac{2}{2-γ}(i + \iT) - (i + \iT) - \iA
            ≥ \tfrac{γ}{2-γ}(i + \iT) - |\BA| ≥ \tfrac{γ}{2-γ}i - 𝒪(\log² n) = Ω(i).
        \end{align*}
        In the end we used the case assumption and a bound of $𝒪(\log n)$ on both the size and the number of native anomalies that hold whp by \cref{lem:canonical-anomalies} \itemref{it:anomaly-size} and \cref{lem:logn-anomalies}.
        •[Late rounds with $i = \tilde{𝒪}(1)$.] We proceed similarly:
        \begin{align*}
            \iO &\itemrefrel{eqit:def-i}{=} i - \iT - \iA
            = (2i - \iA) - i - \iT
            \itemrefrel{eqit:late}{>} k|S_I \setminus \SA| - i - \iT
            \itemrefrel{eqit:degree}{≥} (i + \iT) - (i + \iT) = 0.
        \end{align*}
    \end{description}
    The fact that $\iO > 0$ holds in both cases guarantees a bucket $b ∉ \IA$ storing one key. In the next round of \op{decode'} at least this bucket will be cleared of its key. Progress only stops when $S_r' \setminus \SA = ∅$, i.e.\ when our goal $S_r' ⊆ \SA$ is reached. Concerning the number of rounds we have
    \begin{itemize}
        • $R = 𝒪(1)$ \emph{early} rounds.
        • $𝒪(\log n)$ intermediate rounds since $\iO = Ω(i)$ actually guarantees that a constant fraction of the buckets in $I$ are cleared by the round.
        • $\tilde{𝒪}(1)$ late rounds since each clears at least one of the $\tilde{O}(1)$ remaining buckets from $I$.
    \end{itemize}
    The sum is $\tilde{𝒪}(1)$ rounds as claimed.
\end{proof}

\ifdefined\SIAMmacros
    \bibliographystyle{alpha}
\else
    \medskip
    \noindent
    {\bf Acknowledgement.} Jakob Bæk Tejs Houen and Rasmus Pagh are part of Basic Algorithm Research Copenhagen (BARC), supported by VILLUM Foundation grant 16582. Stefan Walzer is supported by DFG grant WA 5025/1-1.

    \bibliographystyle{plainurl}
\fi
\bibliography{bibliographie}

\end{document}